\documentclass{article}
\usepackage{spconf}
\usepackage[utf8]{inputenc}
\usepackage[T1]{fontenc}
\usepackage[english]{babel}
\usepackage[pdftex]{graphicx}
\usepackage{url}
\usepackage{multirow}
\usepackage{amssymb,amsmath,nccmath,amsthm}
\usepackage{caption}
\usepackage{bm}
\usepackage[noend]{algpseudocode}
\usepackage{algpseudocode}
\usepackage{algorithm}
\usepackage{array}
\usepackage[export]{adjustbox}
\usepackage{cite}
\usepackage{siunitx}
\usepackage{multirow}
\usepackage{enumitem}
\usepackage{pgfplots}
\usepackage{setspace}
\usepackage[colorlinks,citecolor=violet, linkcolor=black]{hyperref}
\usepackage{multibib}
\newcites{apndx}{References}

\usepgfplotslibrary{fillbetween}
\pgfplotsset{every tick label/.append style={font=\scriptsize}}
\def\shortestskip{\setlength{\abovedisplayskip}{0pt}%
\setlength{\belowdisplayskip}{0pt}%
\setlength{\abovedisplayshortskip}{0pt}%
\setlength{\belowdisplayshortskip}{0pt}}

\let\oldselectfont\selectfont
\def\selectfont{\oldselectfont\shortestskip}

\usepackage{etoolbox}
\patchcmd{\thebibliography}
  {\settowidth}
  {\setlength{\itemsep}{0pt plus 0.1pt}\settowidth}
  {}{}
\apptocmd{\thebibliography}
  {\small}
  {}{}

\shortestskip

\pgfplotsset{compat=newest,
       colormap={parula}{
            rgb255=(53,42,135)
            rgb255=(15,92,221)
            rgb255=(18,125,216)
            rgb255=(7,156,207)
            rgb255=(21,177,180)
            rgb255=(89,189,140)
            rgb255=(165,190,107)
            rgb255=(225,185,82)
            rgb255=(252,206,46)
            rgb255=(249,251,14)
        },
    }
\usetikzlibrary{plotmarks}
\usetikzlibrary{arrows.meta}
\usetikzlibrary{spy}
\usepackage{booktabs}
\usepackage{soul}
\makeatletter

\newtheorem{proposition}{Proposition}

\def \ef{\eta(f)}
\def \emf{\eta_m(f)}

\DeclareMathOperator*{\argmax}{argmax}
\def\*#1{\mathbf{#1}}

\makeatletter
\newenvironment{tsubarray}[1]{%
  \vcenter\bgroup
  \Let@ \restore@math@cr \default@tag
  \baselineskip\fontdimen10 \scriptfont\tw@
  \advance\baselineskip\fontdimen12 \scriptfont\tw@
  \lineskip\thr@@\fontdimen8 \scriptfont\thr@@
  \lineskiplimit\lineskip
  \check@mathfonts
  \ialign\bgroup\ifx c#1\hfil\fi
    \normalfont\fontsize\sf@size\z@\selectfont\ignorespaces##\unskip\hfil\crcr
}{%
  \crcr\egroup\egroup
}
\makeatother

\newcommand{\di}{\textbf{(i)} }
\newcommand{\ii}{\textbf{(ii)} }
\newcommand{\iii}{\textbf{(iii)} }
\newcommand{\iv}{\textbf{(iv)} }
\newcommand{\dv}{\textbf{(v)} }
\newcommand{\sga}{SG-v1 }

\newcommand{\ga}{G-v1 }

\begin{document}

\title{Scalable Learning-Based Sampling Optimization for\\  Compressive Dynamic MRI} 
\name{Thomas Sanchez\textsuperscript{1}, Baran G\"ozc\"u\textsuperscript{1}, Ruud B. van Heeswijk\textsuperscript{2}, Armin Eftekhari\textsuperscript{3}, \thanks{This work has received funding from the European Research Council (ERC) under the European Union's Horizon 2020 research and innovation program (grant agreement n$^{\circ}$~725594~-~time-data), from Hasler Foundation Program: Cyber Human Systems (project number 16066) and from the Office of Naval Research (ONR)  (grant n$^\circ$ N62909-17-1-2111).}}
\secondlinename{Efe Il{\i}cak\textsuperscript{4}, Tolga \c{C}ukur\textsuperscript{4}, and Volkan Cevher\textsuperscript{1}}

\address{\textsuperscript{1}EPFL, Switzerland
     \textsuperscript{2}CHUV, Switzerland
     \textsuperscript{3}Ume\aa~University, Sweden
    \textsuperscript{4}Bilkent University, Turkey\vspace{-.5cm}}

\maketitle

\begin{abstract}

Compressed sensing applied to magnetic resonance imaging (MRI) allows to reduce the scanning time by enabling images to be reconstructed from highly undersampled data. In this paper, we tackle the problem of designing a sampling mask for an arbitrary reconstruction method and a limited acquisition budget. Namely, we look for an optimal probability distribution from which a mask with a fixed cardinality is drawn. We demonstrate that this problem admits a compactly supported solution, which leads to a deterministic optimal sampling mask.  We then propose a stochastic greedy algorithm that \di provides an approximate solution to this problem, and \ii resolves the scaling issues of \cite{gozcu2017, gozcu2019rethinking}.  We validate its performance on \textit{in vivo} dynamic MRI with retrospective undersampling, showing that our method preserves the performance of \cite{gozcu2017, gozcu2019rethinking} while reducing the computational burden by a factor close to 200. Our implementation is available at
\href{https://github.com/t-sanchez/stochasticGreedyMRI}{https://github.com/t-sanchez/stochasticGreedyMRI}.

\end{abstract}

\begin{keywords}
    Magnetic resonance imaging, compressive sensing (CS), learning-based sampling.
\end{keywords}

\vspace{-4mm}
\section{Introduction}
\vspace{-.3cm}
Dynamic Magnetic Resonance Imaging (dMRI) is a powerful tool in medical imaging, which allows for non-invasive monitoring of tissues over time. A main challenge to the quality of dMRI examinations is the inefficiency of data acquisition that limits temporal and spatial resolutions. In the presence of moving tissues, such as in cardiac MRI, the trade-off between spatial and temporal resolution is further complicated by the need to perform breath-holds to minimize motion artifacts \cite{saeed2015cardiac}.

In the last decade, the rise of Compressed Sensing (CS) has significantly contributed to overcoming these problems.  CS allows for a successful reconstruction from undersampled measurements, provided that they are incoherent \cite{candes2006stable,donoho2006compressed} and that the data can be sparsely represented in some domain. In dMRI, samples are acquired in the $k$-$t$ space (spatial frequency and time domain), and can be sparsely represented in the $x$-$f$ domain (image and temporal Fourier transform domain). 
Many algorithms have exploited this framework with great success (see \cite{lustig2006kt,jung2009k,otazo2010combination,lingala2011accelerated, feng2014golden, caballero2014dictionary, otazo2015low, jin2016general, schlemper2018deep} and the references therein). 

While CS theory mostly focuses on fully random measurements \cite{candes2006robust}, the practical implementations have generally exploited \textit{random variable-density sampling}, based on drawing random samples from a parametric distribution (typically polynomial or Gaussian) which reasonably imitates the energy distribution in the $k$-$t$ space \cite{lustig2007sparse, jung2007improved}. 
 While all these approaches allow to quickly design masks which yield a great improvement over fully random sampling, prescribed by the theory of CS, they \di remain largely heuristic; \ii ignore the anatomy of interest; \iii ignore the reconstruction algorithm; \iv require careful tuning of their various parameters, and \dv do not necessarily use a fixed number of readouts per frame.

In the present work, we show that the problem of finding an optimal mask sampling distribution which contains $n$ out of $p$ possible locations admits a solution compactly supported on $n$ elements. This demonstrates that our previously proposed framework in \cite{gozcu2017, gozcu2019rethinking}, which searches for an approximately optimal sampling mask, is in fact looking for a solution to the more general problem of finding an optimal measurement distribution.
In addition, we propose a \textit{scalable learning-based framework} for dMRI. Our proposed stochastic greedy method preserves the performance of  \cite{gozcu2017, gozcu2019rethinking} while reducing the computational burden by a factor close to $200$.  


Numerical evidence shows that our framework can successfully find sampling patterns for a broad range of decoders, from k-t FOCUSS \cite{jung2009k} to ALOHA \cite{jin2016general}, outperforming state-of-the-art model-based sampling methods over nearly all sampling rates considered. 


\vspace{-4mm}
\section{Theory}\label{sec:background}\shortestskip
\vspace{-.3cm}

\subsection{Signal Acquisition} \label{sec:pbm}
In the compressed sensing (CS) problem \cite{donoho2006compressed}, one desires to retrieve a signal that is known to be sparse in some basis using only a small number of linear measurements. In the case of dynamic MRI, we consider a signal $\*x \in \mathbb{C}^{p} = \mathbb{C}^{N^2 T}$ (i.e. a vectorized video of size $N\times N$ with $T$ frames), and the subsampled Fourier measurements are \useshortskip
\begin{equation}
\*b = \*P_{\Omega}\*{\Psi} \*x + \*w \label{eq:samp}
\end{equation}
where $\*{\Psi} \in \mathbb{C}^{p}$ is the spatial Fourier transform operator applied to the vectorized signal, $\*P_{\Omega} : \mathbb{C}^{p} \to \mathbb{C}^{n}$ is a subsampling operator that  
selects the rows of $\*\Psi$ according to the indices in the set $\Omega$ with $|\Omega| = n$ and $n \ll p$. We refer to $\Omega$ as \textit{sampling pattern} or \textit{mask}. We assume the signal $\*x$ to be sparse in the basis $\*\Phi$, which typically is a temporal Fourier transform across frames. Given the samples $\*b$, along with $\Omega$, a \textit{reconstruction algorithm} or \textit{decoder} $g$ forms an estimate $\*{\hat{x}}  = g(\*b, \Omega) $ of $\*x$.

The quality of the reconstruction is then evaluated using a \textit{performance metric} $\eta(\*x,\bm{\hat{\*x}})$, which could typically include Peak Signal-to-Noise Ratio (PSNR), the negative Mean Square Error (MSE), or the Structural Similarity Index Measure (SSIM) \cite{wang2004image}. 

\vspace{-.5cm}
\subsection{Sampling mask design}
\vspace{-.2cm}

We model the mask designing process as finding a probability mass function (PMF) $f \in S^{p-1}$, where $S^{p-1} := \{f \in [0,1]^p : \sum_{i=1}^p f_i =1\}$ is the standard simplex in $\mathbb{R}^p$. $f$ assigns to each location $i$ in the $k$-space a probability $f_i$ to be acquired. The mask is then constructed by drawing without replacement from $f$ until the cardinality constraint $|\Omega|=n$ is met. The problem of finding the optimal sampling distribution is subsequently formulated as \useshortskip
\begin{equation}
\max_{f\in S^{p-1}} \ef, 
\qquad \ef :=  \mathbb{E}_{\substack{\Omega(f,n)\\ \*x \sim \mathcal{P}_{\*x}}}\left[\eta\left(\*x, \*{\hat{x}}\left(\Omega,\*x\right)\right)\right],
\label{eq:main}
\end{equation}
where the index set $\Omega\subset [p] $ is generated from $f$ and $[p] := \{1,\ldots,p\}$. This problem corresponds to finding the probability distribution $f$ that maximizes the expected performance metric with respect to the data $\mathcal{P}_{\*x}$ and the masks drawn from this distribution. To ease the notation, we will use $\eta\left(\*x, \*{\hat{x}}\left(\Omega,\*x\right)\right) \equiv \eta\left(\*x; \Omega\right)$.


In practice, we do not have access to $\mathbb{E}_{\mathcal{P}_{\*x}} \left[\eta(\*x; \Omega)\right]$ and instead have at hand the training images $\{\*x_i\}_{i=1}^m$ drawn independently from $\mathcal{P}_{\*x}$. We therefore maximize the empirical perfromance by solving \useshortskip
\begin{equation}
\label{eq:emp}
\max_{f\in S^{p-1}}  \emf, \text{~~~~} \emf :=\frac{1}{m} \sum_{i=1}^m \mathbb{E}_{\Omega(f,n)}\left[\eta(\Omega,\*x_i)\right].
 \end{equation}\useshortskip 
Given that Problem \eqref{eq:emp} looks for  masks that are constructed by sampling $n$ times without replacement from $f$, the following holds.
\begin{proposition}
There exists a maximizer of Problem \eqref{eq:emp} that is supported on an index set of size at most $n$.\label{prop:1}
\end{proposition}
\begin{proof}
Let the distribution $\widehat{f}_n$ be a maximizer of Problem~\eqref{eq:emp}. We are interested in finding the support of $\widehat{f}_n$. Because $\sum_{|\Omega|=n}\Pr[\Omega]=1$, note that 

\begin{align}
\max_{f\in S^{p-1}}  \emf & := \max_{f\in S^{p-1}}  \sum_{|\Omega|=n}  \frac{1}{m}\sum\nolimits_{i=1}^m\eta(\*x_{i};\Omega) \cdot \Pr[\Omega|f]\nonumber\\ 
& \le \max_{f\in S^{p-1}} \max_{|\Omega|=n} \frac{1}{m}\sum\nolimits_{i=1}^m\eta(\*x_{i}; \Omega) \nonumber\\
& = \max_{|\Omega|=n}   \frac{1}{m}\sum\nolimits_{i=1}^m\eta(\*x_{i};\Omega). 
\end{align}
Let $\widehat{\Omega}_n$ be an index set of size $n$ that maximizes the last line above. 
The above holds with equality when $\Pr[\widehat{\Omega}_n]=1$ and $\Pr[\Omega]=0$ for $\Omega\ne \widehat{\Omega}_n$ and $f=\widehat{f}_n$. This in turn happens when $\widehat{f}_n$ is  supported on $\widehat{\Omega}$. That is, there exists a maximizer of Problem \eqref{eq:emp} that is supported on an index set of size $n$. 
\end{proof}
While this observation does not indicate how to find this maximizer, it nonetheless allows us to further simplify Problem \eqref{eq:emp}. More specifically, the observation that a distribution $\widehat{f}_n$ has a compact support of size $n$ implies the following:
\begin{proposition}\useshortskip
\begin{equation}
\text{Problem \eqref{eq:emp}} \equiv  \max_{|\Omega|=n}  \frac{1}{m}\sum_{i=1}^m \eta(\*x_{i}; \Omega) \label{eq:greedy}
\end{equation} \label{prop:2}
\vspace{-.5cm}
\end{proposition}
\begin{proof}
Proposition~\ref{prop:1} tells us that a solution of Problem~\eqref{eq:emp} is supported on a set of size at most $n$, which implies
\begin{equation}
\text{Problem \eqref{eq:emp}} \equiv 
\max_{f\in S^{p-1}, |\text{supp}(f)| = n }\emf\label{eq:main emp equiv 1}.
\end{equation} 
That is, we only need to search over compactly supported distributions $f$. Let $S_\Gamma$ denote the standard simplex on a support $\Gamma\subset [p]$. It holds that
\begin{align}
\text{Problem \eqref{eq:main emp equiv 1}} &  \equiv
\max_{|\Gamma|=n} \max_{f\in S_\Gamma}\emf \nonumber\\
 = \max_{|\Gamma|=n} &\max_{f \in S_\Gamma} \frac{1}{m}\sum\nolimits_{i=1}^m\eta(\*x_{i};\Gamma) \cdot \Pr[\Gamma|f] \nonumber\\
 = \max_{|\Gamma|=n}& \max_{f \in S_\Gamma} \frac{1}{m}\sum\nolimits_{i=1}^m\eta(\*x_{i};\Gamma) \nonumber\\
 =  \max_{|\Gamma|=n} & \frac{1}{m}\sum\nolimits_{i=1}^m\eta(\*x_{i};\Gamma).
\label{eq:main emp equiv 2}
\end{align}
To obtain the second and third equalities, one observes that all masks have a common support $\Gamma$ with $n$ elements, i.e. $f\in S_\Gamma$ allows only for a single mask $\Omega$ with  $n$ elements, namely $\Omega=\Gamma$.~\end{proof}
The framework of Problem~\eqref{eq:emp} captures most variable-density based approaches of the literature that are defined in a data-driven fashion \cite{seeger2010optimization,ravishankar2011adaptive,vellagoundar2015robust,weizman2015compressed,haldar2019oedipus,bahadir2019learning,sherry2019learning}, and Proposition \ref{prop:2} shows that Problem~\eqref{eq:main emp equiv 2}, that we tackled in \cite{gozcu2017,gozcu2019rethinking} and develop here, also aims at solving the \textit{same} problem as these probabilistic approaches. 
Note that while the present theory considered sampling \textit{points} in the Fourier space, it is readily applicable to the Cartesian case, where full lines are added to the mask at once. 

\vspace{-.3cm}
\section{Stochastic greedy mask design}\label{sec:learning}
\vspace{-.3cm}
Aligned with the approach that we previously proposed in \cite{gozcu2017}, we want to find an approximate solution to Problem~\eqref{eq:greedy} by leveraging a greedy algorithm. This is required by Problem~\eqref{eq:greedy} being inherently combinatorial. The previous greedy method of \cite{gozcu2017,gozcu2019rethinking} suffers from three main drawbacks: \di it scales quadratically with the total number of lines, \ii it scales linearly with the size of the dataset, and \iii it does not construct mask with a fixed number of readouts by frame. While \cite{gozcu2019rethinking} partially deals with \textbf{(i)}, our proposed stochastic greedy approach addresses all three issues, while preserving the benefits of \cite{gozcu2017}. It notably still preserves the nestedness and ordering of the acquisition, where critical locations are acquired initially, and the mask built outputs a nested structure (i.e. the mask at $30\%$ sampling rate includes all sampling locations of the mask at $20\%$).

Let us introduce the set $\mathcal S$ of all lines that can be acquired, which is a set of subsets of $\{1,\dotsc,p\}$. A feasible Cartesian mask takes the form $ \Omega = \bigcup_{j=1}^\ell S_j, \quad S_j \in \mathcal{S}$, i.e. it consists of a union of lines.  Both the greedy method of \cite{gozcu2017} and our stochastic method are detailed in Algorithm \ref{alg:1} below. Our stochastic greedy method (\textbf{SG-v2}) addresses the three main limitations of the greedy method of \cite{gozcu2017} (\textbf{G-v1}).  The issue \di is solved by picking uniformly at random  at each iteration a batch possible lines $\mathcal{S}_{iter}$ of size $k$  from a given frame $\mathcal{S}_t$, instead of considering the full set of possible lines $\mathcal{S}$ (line 3 in Alg. \ref{alg:1}); \ii is addressed by considering  a fixed batch of training data $\mathcal{L}$ of size $l$ instead of the whole training set of size $m$ at each iteration (line 4 in Alg. \ref{alg:1}); \iii is solved by iterating through the lines to be added from each frame $\mathcal{S}_t$ sequentially (lines 1, 3 and 10 in Alg. \ref{alg:1}).  These improvements are inspired by the refinements done to the standard greedy algorithm in the field of submodular optimization \cite{mirzasoleiman2015lazier}, and allow to move the computational complexity from $\Theta\left(mr(NT)^2\right)$ to $\Theta\left(lrkNT\right)$, effectively speeding up the computation by a factor $\mathbf{\Theta(\frac{m}{l}\frac{NT}{k})}$. Our results show that this is achieved without sacrificing any reconstruction quality.\vspace{-.2cm}

\begin{algorithm}[ht]
\caption{Greedy mask optimization algorithms for dMRI \\ \textbf{(G)} refers to the greedy algorithm \cite{gozcu2017} \\   \textbf{(SG)} refers to the \textbf{stochastic greedy} algorithm  \\ \textbf{(v1)}  algorithm iterated throughout the whole training set \\ \textbf{(v2)} algorithm iterated through batches of training examples}\label{alg:1}
\textbf{Input}: Training data $\{\*x\}_{i=1}^m$, recon. rule $g$, sampling set $\mathcal{S}$, max. cardinality $n$, samp. batch size $k$, train. batch size $l$ \\
\textbf{Output}: Sampling pattern $\Omega$
\begin{algorithmic}[1]
\State 
 \begin{tabular}{ccc} 
\textbf{(SG) }  Initialize $t=1$  
    \end{tabular}
 \While{$|\Omega| \leq n$}
        
        \State   $ \left \{\begin{tabular}{l}
 \textbf{(G)~~~} Pick $\mathcal{S}_{iter} = \mathcal{S}$   \\
 \textbf{(SG)~} Pick $\mathcal{S}_{iter}\subseteq \mathcal{S}_t$ at random, with $|\mathcal{S}_{iter}| = k$  
  \end{tabular} \right.$
        \State   $ \left \{\begin{tabular}{l}
 \textbf{(v1) } Pick $\mathcal{L} = \{1,\ldots,m\}  $ \\
 \textbf{(v2) } Pick $\mathcal{L} \subseteq \{1,\ldots,m\} $, with $|\mathcal{L}| = l $
  \end{tabular} \right.$
         \For{$S \in \mathcal{S}_{iter} \text{ such that } |\Omega \cup S| \leq \Gamma$} 
     
        \State  $\Omega' = \Omega \cup S$ 
        \State For each $\ell \in \mathcal{L}$ set ${\hat{\*x}}_\ell\leftarrow g(\Omega', \*P_{\Omega'}\bm\Psi\*x_\ell)$
        \State   $\eta(\Omega') \leftarrow \frac{1}{|\mathcal{L}|} \sum_{\ell\in\mathcal{L}} \eta(\*x_\ell, {\hat{\*x}}_\ell)$
 \EndFor
 \State $\displaystyle\Omega \leftarrow \Omega \cup S^*,  \text{ where }$   $\displaystyle S^* = \argmax_{\substack{S:|\Omega \cup S| \leq n}} \eta(\Omega\cup S)$
\State   \textbf{ (SG) }   $t= (t \bmod T)+1$  
    \EndWhile
\State {\bf return} $\Omega$ 
\end{algorithmic}

\end{algorithm} 
\vspace{-.5cm}
\section{Numerical Experiments}\label{sec:experiments}
\vspace{-.3cm}
\subsection{Implementation details}
\vspace{-.2cm}
\textbf{Reconstruction algorithms:}  We consider three reconstruction algorithms, namely \textit{$k$-$t$ FOCUSS} (KTF) \cite{jung2009k},  and \textit{ALOHA} \cite{jin2016general}.  Their parameters were selected to maintain a good empirical performance across all sampling rates considered.

\textbf{Mask selection baselines:}\vspace{-.3cm}

\begin{itemize}[leftmargin=*]
\item\textit{Coherence-VD} \cite{lustig2007sparse}: We consider a random \textit{variable-density} sampling mask with Gaussian density and optimize its parameters to minimize coherence.
\item\textit{LB-VD} \cite{gozcu2017,gozcu2019rethinking}: Instead of minimizing the coherence as in \textit{Coherence-VD}, we perform a grid search on the parameters using the training set to optimize reconstruction according to the same performance metric as our method.
\end{itemize}
\textbf{Data sets:}  Our dynamic data were acquired in seven adult volunteers with a balanced steady-state free precession (bSSFP) pulse sequence on a whole-body Siemens 3T scanner using a 34-element matrix coil array. Several short-axis cine images were obtained during a breath-hold scan. Fully sampled Cartesian data were acquired using a $256\times 256$ grid with $25$ frames, then combined and cropped to  a $152 \times 152 \times 17$ single coil image. The details of the parameters used are provided in the supplementary material \cite{sanchez2018}. In the experiments, we used three volumes for training and four for testing. 

\vspace{-.4cm}
\subsection{Comparison of greedy algorithms}
\vspace{-.2cm}
We first compare the performance of \ga with SG-v1 and SG-v2, and show the results on Figure \ref{fig:psnr_batch}. We are specifically interested in determining the sensitivity of our algorithm to the sampling batch size $k$ and training batch size $l$ (for SG-v2, we use $l=1$ unless stated differently). We see that using a small batch size $k$ (e.g. $10$) yields a drop in performance, while $k=38$ even improves performance compared to G-v1, with respectively $60$ times less computation for \sga and $180$ less computations for SG-v2. One should also note that using a batch of training images (SG-v2) does not reduce the performance compared to SG-v1, while largely reducing computations. Also, additional results (in the supplementary material \cite{sanchez2018}) show that using larger batches yields similar results as for $k=38$. The fact that the performance of SG-v2 with $k=38$ outperforms G-v1 could be surprising, but originates in the lack of structure of the problem, where introducing noise in the computations  through random batches of samples improves the overall performance of the method. In the sequel, we use $k=38$ and $l=1$ for SCG-v2.

\begin{figure}[!tbp]

  \centering
  \begin{minipage}[b]{.58\linewidth}
 \hspace{-.3cm}\includegraphics[width=\linewidth]{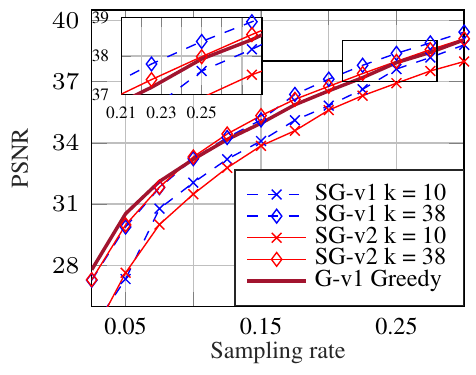}
\end{minipage}\hfill
 \begin{minipage}[b]{.4\linewidth}
 \caption{PSNR as a function of the sampling rate for KTF, comparing the different reconstruction methods as well as the effect of the batch size on the quality of the reconstruction for SG. }\label{fig:psnr_batch}
 \end{minipage} 
  \vspace{-.8cm}
\end{figure}
\vspace{-.4cm}
\subsection{Single coil results}\label{sec:baseline_main} 
\vspace{-.2cm}
The comparison to baselines is shown on Figures \ref{fig:psnr_cross} and \ref{fig:plot_cross}, where we see that the SG-v2 method yields masks that consistently improve the results compared to all variable-density methods used.

We notice in Figure \ref{fig:plot_cross} that comparing the reconstruction algorithms with VD methods do not allow for a faithful performance comparison of the reconstruction algorithms: the performance difference is very small between the reconstruction methods. In contrast, considering the reconstruction algorithm jointly with a sampling pattern optimized with our model-free approach makes the performance difference much more noticeable: ALOHA with its corresponding mask clearly outperforms KTF, and this conclusion could not be made by looking solely at reconstructions with VD-based masks. Note that extended results, along with multi-coil experiments, are available in our supplementary material \cite{sanchez2018}.

\begin{figure}[!t]
\centering
\vspace{-.7cm}
\hspace{-.3cm}\includegraphics[width=.95\linewidth]{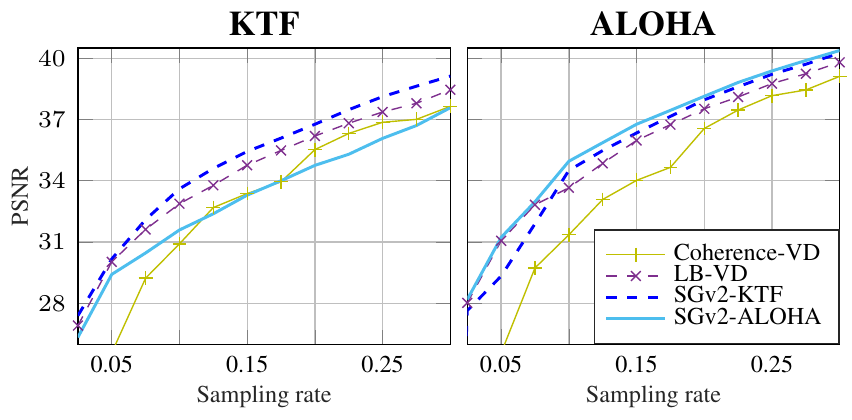}
\vspace{-.5cm}
\caption{PSNR as a function of sampling rate for both reconstruction algorithms considered, comparing the mask design methods considered, averaged on $4$ images.}\label{fig:psnr_cross}
\vspace{-.5cm}
\end{figure}

\begin{figure}[!t]
\centering
\vspace{-.7cm}
\includegraphics[width=.95\linewidth]{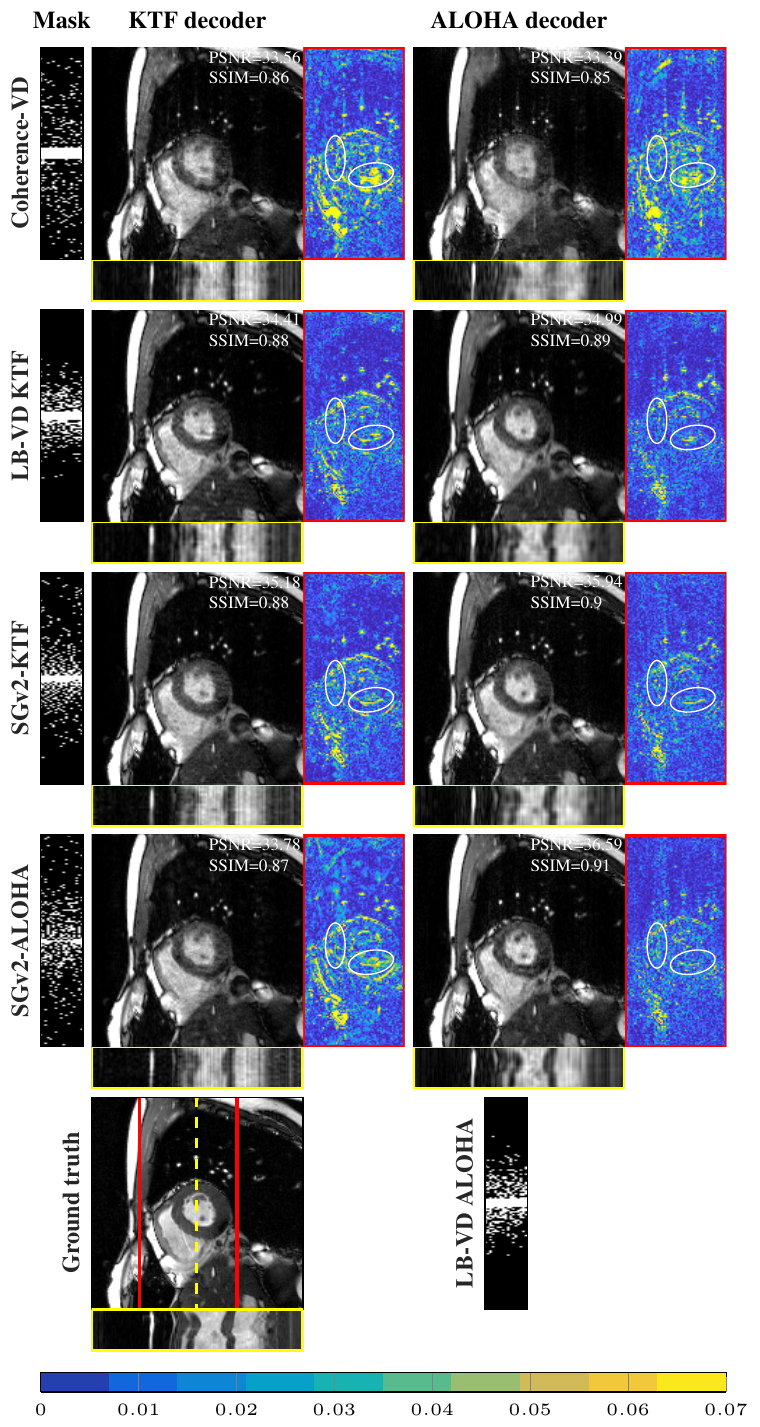}
\vspace{-.4cm}
\caption{Comparison of the different reconstruction masks and decoders, for a sampling rate of $15\%$ on a single sample with its PSNR/SSIM performances.}\label{fig:plot_cross}
\vspace{-.5cm}
\end{figure}
\vspace{-.5cm}
\subsection{Large scale static results}
\vspace{-.3cm}

This last experiment shows the scalability of our method to very large datasets. We used the fastMRI dataset \cite{zbontar2018fastmri} consisting of knee volumes  and trained the mask for reconstructing the $13$ most central slices of size $320\times 320$, which yielded a training set containing $12649$ slices. For the sake of brevity, we only report computations performed using total variation (TV) minimization with NESTA \cite{becker2011nesta}. For mask design, we used the SG-v2 method with $k=80$ and $l=20$ (2500 fewer computations compared to G-v1). The LB-VD method was trained using $80$ representative slices and optimizing the parameters with a similar computational budget as SG-v2. The result on Figure \ref{fig:psnr_fastmri} shows a uniform improvement of our method over the LB-VD approach.

\begin{figure}[!t]
  \centering
  \begin{minipage}[t]{.58\linewidth}
  \vspace{0pt}
 \hspace{-.3cm}\includegraphics[width=\linewidth]{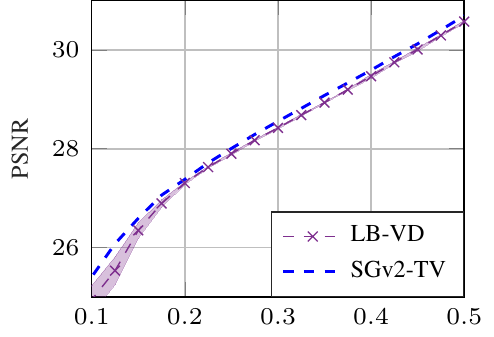}
\end{minipage}\hfill
 \begin{minipage}[t]{.42\linewidth}
 \vspace{0pt}
 \caption{PSNR as a function of the sampling rate for TV, averaged on the 13 most central slices of the fastMRI validation set \cite{zbontar2018fastmri} (2587 slices). SGv2 outperforms LB-VD over all sampling rates.}\label{fig:psnr_fastmri}
  \vspace{-1cm}
 \end{minipage} 
   \vspace{-.8cm}
\end{figure}

\vspace{-.3cm}
\section{Discussion and Conclusion}\label{sec:conclusion}
\vspace{-.3cm}
We presented a scalable sampling optimization method for dMRI, which largely addresses the scalability issues of  \cite{gozcu2017,gozcu2019rethinking}.  Reducing the resources used by \ga by as much as a $200$ times was shown to have no negative impact on the quality of reconstruction achieved within our framework. Our method was demonstrated to successfully scale to very large datasets such as fastMRI \cite{zbontar2018fastmri}, which the previous greedy method \cite{gozcu2017} could not achieve.

The masks obtained bring significant image quality improvements over the baselines. The results suggest that VD-based methods limit the performance of CS applied to MRI through their underlying model. They are consistently outperformed by our model-free and data-adaptive method on different \textit{in vivo} datasets, across several decoders, field of views and resolutions. Our findings highlight that sampling design should not be considered in isolation from data and reconstruction algorithm, as using a mask that is not specifically optimized can considerably hinder the performance of the algorithm.

More importantly, our theoretical results show that the generic non-convex Problem~\eqref{eq:emp} aiming at finding a probability mass  function under a cardinality constraint from which a mask is subsequently sampled, is equivalent to the discrete Problem~\eqref{eq:main emp equiv 2} of looking for the support of this PMF. This connection opens the door to rigorously leveraging techniques from combinatorial optimization for the problem of designing optimal, data-driven sampling masks for MRI.

\vspace{-1.2cm}
\bibliographystyle{IEEEtran}
\vspace{-1.2cm}
\bibliography{biblio}
\appendix
\newpage
\section{Detailed description of the datasets}
\textbf{Cardiac dataset.} The  data set was acquired in seven healthy adult volunteers with a balanced steady-state free precession (bSSFP) pulse sequence on a whole-body Siemens 3T scanner using a 34-element matrix coil array. Several short-axis cine images were acquired during a breath-hold scan. Fully sampled Cartesian data were acquired using a $256\times 256$ grid, with relevant imaging parameters including \SI{320 x 320}{\milli\meter} field of view (FoV), \SI{6}{\milli\meter} slice thickness, \SI{1.37 x 1.37}{\milli\meter} spatial resolution, \SI{42.38}{\milli\second} temporal resolution, \SI{1.63/3.26}{\milli\second} TE/TR, \ang{36} flip angle, \SI{1395}{\hertz}\mbox{/px} readout bandwidth. There were $13$ phase encodes acquired for a frame during one heartbeat, for a total of $25$ frames after the scan. 

The Cartesian cardiac scans were then combined to single coil data from the initial $256\times 256 \times 25 \times 34$ size, using adaptive coil combination \citeapndx{walsh2000adaptive, griswold2002use}, which keeps the image complex. This single coil image was then cropped to a $152 \times 152 \times 17$ image. This is done because a large portion of the periphery of the images are static or void, and also to enable a greater computational efficiency. 

\noindent
\textbf{Vocal dataset.} The vocal dataset that we used in the experiments \ref{sec:single} comprised $4$ vocal tract scans with a 2D HASTE sequence (T2 weighted single-shot turbo spin-echo) on a 3T Siemens Tim Trio using a 4-channel body matrix coil array. The study  was  approved  by  the  local  institutional  review  board,  and  informed  consent  was  obtained  from  all  subjects  prior  to  imaging.  Fully sampled Cartesian data were acquired using a $256 \times 256$ grid, with \SI{256 x 256}{\milli\meter} field of view (FoV), \SI{5}{\milli\meter} slice thickness, \SI{1 x 1}{\milli\meter} spatial resolution, \SI{98/1000}{\milli\second} TE/TR, \ang{150} flip angle,  \SI{391}{\hertz}\mbox{/px} readout bandwidth, \SI{5.44}{\milli\second} echo spacing ($256$ turbo factor). There was a total of $10$ frames acquired, which were recombined to single coil data using adaptive coil combination as well \citeapndx{walsh2000adaptive, griswold2002use}.  

\noindent
\textbf{fastMRI.} The fastMRI dataset was obtained from the NYU fastMRI initiative \cite{zbontar2018fastmri}. The anonymized dataset comprises raw k-space data from more than 1,500 fully sampled knee MRIs obtained on 3 and 1.5 Tesla magnets. The dataset includes coronal proton density-weighted images with and without fat suppression.

\vspace{-.4cm}
\section{Extended literature review}

The most widely used approach for the design of the sampling pattern $\*\Omega$ is \textit{random variable-density sampling}, which was originally proposed by Lustig et al. \cite{lustig2007sparse} for static MRI and adapted to dynamic MRI by Jung et al. \cite{jung2007improved}. It offers a compromise between incoherent measurements, required by the theory of CS, and the structure that can be found in the \mbox{k-space}, where most of the energy is concentrated in the low frequency end of the spectrum. This classical approach draws random samples according to a parametric distribution mimicking the energy distribution of the \mbox{k-space}, favoring low-frequency samples. The distribution considered is typically either polynomial \cite{lustig2007sparse, weizman2015compressed} \citeapndx{kim2012accelerated,tremoulheac2014dynamic}, or Gaussian \cite{jung2009k, otazo2010combination, caballero2014dictionary,otazo2015low,jin2016general,schlemper2018deep}. In these setups, a slight offset is often added in order to prevent the distribution from having extremely small probabilities at high-frequencies, and a few low-frequency \mbox{k-space} samples are acquired at the Nyquist rate.

The variable-density based methods commonly used in dMRI perform well, but have several weaknesses, already highlighted in \cite{gozcu2017} for static MRI. They require parameters to be tuned, such as decay rate of the polynomial, the standard deviation of the Gaussian distribution or the number of central phase encodes and arbitrarily constrain the sampling patterns to a model without any theoretical justification. Moreover, it is unclear which sampling density will be most effective for a given anatomy and reconstruction rule. Also, the idea of randomizing the acquisition is in itself questionable, as in practice, one would desire to design a fixed sampling pattern that we will know to perform well for a specific anatomy across many subjects. Finally, some variable-density methods, such as Poisson Disc Sampling \citeapndx{vasanawala2011practical}, do not use a fixed number of readouts per frame, which complicates their hardware implementation for dynamic MRI \citeapndx{ahmad2015variable}. Indeed, undersampling some frames more heavily than others might result in missing critical temporal information.

Recently, several articles have focused on improved design of spatiotemporal sampling patterns for dMRI, and we hereafter detail two particularly relevant methods. A recent method devised for this purpose is the variable density incoherent spatiotemporal acquisition (VISTA) \citeapndx{ahmad2015variable} that maximizes Riesz energy on a spatiotemporal grid, and has the notable advantage of generating patterns with high levels of incoherence, and maintaining uniform sampling density across frames. Another important technique proposed by Li et al. \citeapndx{li2018dynamic} develops a method for Cartesian sampling exploiting the golden-ratio, with the aim to generate incoherent measurements and maintain uniform sampling density across frames\footnote{This approach is different from the commonly used golden-\textit{angle} sampling used in radial sampling.}. 

Other relevant undersampling works include, in the non-Cartesian  setting, fully random radial sampling \citeapndx{jung2010radial, tremoulheac2014dynamic}, as well as golden-angle radial sampling,  where spokes separated by the golden-angle are continuously acquired \citeapndx{winkelmann2007optimal, feng2014golden,feng2016xd}. These results exploit the inherent advantage of radial over Cartesian sampling that each spoke goes through the sample of the k-space and can thus contain low-frequency as well as high-frequency information. More recent work also leverage variable-density approaches in the non-Cartesian setting \citeapndx{boyer2016generation,lazarus2018variable} Also, in static MRI, several methods exploiting training signals have been proposed: in \citeapndx{knoll2011adapted, zhang2014energy,vellagoundar2015robust}, a distribution from which random samples are drawn is constructed, and in \citeapndx{seeger2010optimization,ravishankar2011adaptive,liu2012under, haldar2019oedipus}, a single image is used at a time to determine the sampling mask. Very recently, deep-learning based methods have enabled active mask design paired with online reconstruction and shown very promising results \citeapndx{jin2019self,zhang2019reducing,weiss2019learning}. However, to the best of our knowledge, none of these methods have been extended to dynamic MRI. 

 \vspace{-.4cm}
\section{Influence of the batch size $k$ on the mask design}\label{app:batch}

In this appendix, we discuss the tuning of the batch size used in SG-v1, to specifically study the effect of different batch sizes. We ran SG-v1 with different batch sizes in the same settings are in the numerical experiment of section \ref{sec:baseline_main} and report on Figure \ref{fig:psnr_batch2} the PSNR of the reconstructions for SG-v1. We only considered KTF for brevity. We see that very small batch sizes yield poor results, and the PSNR reaches the result from G-v1 with as few as 38 samples (out of $152 \times 17 = 2584$ samples overall). Unless then the batch size is extremely small (less than $1$ to $2\%$ of all phase encoding lines at each greedy iteration), the results suggest that the masks obtained with SG-v1 or SG-v2 yield satisfactory reconstruction quality, i.e. the same quality as G-v1 or even an increase. 

The Figure \ref{fig:mask_sto_cycling} shows the different masks obtained for the batch sizes considered, several observations can be made. First of all, as expected, taking a batch size of $1$ yields a totally random mask, and taking a batch size of $5$ yields a mask that is more centered towards low frequency than the one with $k=1$ but it still has a large variance. Then, as the batch size increases, resulting masks seem to converge to very similar designs, but those are slightly different from the ones obtained with G-v1.  

\begin{figure}
  \centering
  \begin{minipage}[b]{.6\linewidth}
 \hspace{-.3cm}\includegraphics[width=\linewidth]{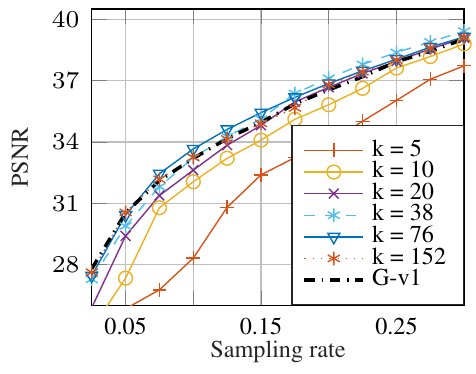}
\end{minipage}\hfill
 \begin{minipage}[b]{.4\linewidth}
\caption{PSNR as a function of the rate for KTF, comparing the effect of the batch size on the quality of the reconstruction for SG-v1. The result is averaged on $4$ testing images of size $152\times 152\times 17$.}\label{fig:psnr_batch2}
 \end{minipage} 
  \vspace{-.8cm}
\end{figure}

\begin{figure}[!ht]
 \vspace{-.4cm}
\centering
	\includegraphics[width=.95\linewidth]{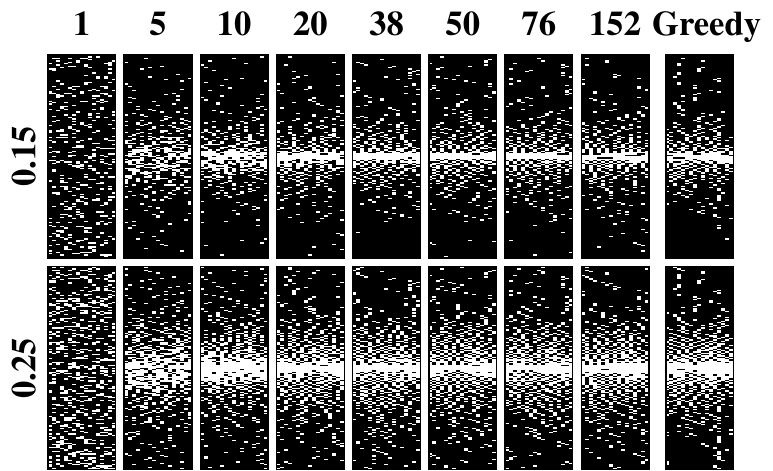}
	\caption{Learning-based masks obtained with SG-v1 for different batch sizes $k$ using KTF as a reconstruction algorithm, shown in the title of each column, for $15\%$ and $25\%$ sampling rate. The optimization used data of size $152\times 152 \times 17$, with a total of $2584$ possible phase encoding lines for the masks to pick from.}\label{fig:mask_sto_cycling}
	 \vspace{-.4cm}
\end{figure}

 \vspace{-.4cm}
\section{Computational costs}
\vspace{-.4cm}

We report here the computational costs for the different variations of the greedy methods used in the single coil experiment \ref{sec:baseline_main} as well as the computational costs for the Appendix \ref{sec:single}. Table \ref{tab:runtime}  provides the running times and empirically measured speedup for the greedy variation, and Table \ref{tab:LB-VD} provides the computational times required to obtain the learning-based variable density (LB-VD) parameters through an extensive grid-search. The empirical speedup is computed as 
\begin{equation}
\text{Speedup} = \frac{t_{\text{G-v1}}\cdot n_{\text{procs, G-v1}}}{t_{\text{SG-v2}}\cdot n_{\text{procs, SG-v2}}}
\end{equation}
The main point of these tables is to show that the computational improvement is very significant in terms of resources, and that our approach improves greatly the efficiency of the method of \cite{gozcu2017}. This ratio might differ from the predicted speedup factor of $\mathbf{\Theta(\frac{m}{l}\frac{NT}{k})}$ due to computational considerations. Table \ref{tab:runtime} shows that we have roughly a factor $1.2$ between the predicted and the measured speedup, mainly due to the communication between the multiple processes as well as I/O operations.

\begin{table*}[!t]
\centering
  \caption{Running time of the greedy algorithms for different decoders and training data sizes. The setting corresponds to  $n_{x}$, $n_{y}$, $n_{\text{frames}}$, $n_{\text{train}}$. $n_{\text{procs}}$ is the number of parallel processes used by each simulation. $^*$ means that the runtime was extrapolated from a few iterations. We used $k=n_{\text{procs}}$ for SG-v1 and SG-v2 and $l=3$ for SG-v2. The speedup column contains the measured speedup and the theoretical speedup in parentheses.}\label{tab:runtime}
  \begin{tabular}{>{\bfseries}c crccrcccc}
    \toprule
     \multirow{2}{*}{\textbf{Algorithm}} &\multirow{2}{*}{\textbf{Setting}} & \multicolumn{2}{c}{\textbf{G-v1}} & \multicolumn{3}{c}{\textbf{SG-v1}} &  \multicolumn{3}{c}{\textbf{SG-v2}} \\
     \cmidrule(r){3-4} \cmidrule(l){5-7} \cmidrule(l){8-10}
    &  &  \multicolumn{1}{c}{Time} &  $n_{\text{procs}}$  &  \multicolumn{1}{c}{Time} & $n_{\text{procs}}$ &  Speedup& \multicolumn{1}{c}{Time} & $n_{\text{procs}}$ &Speedup\\ 
    \midrule

     \multirow{2}{*}{KTF}  &152$\times$152$\times$17$\times$3 & $6\text{d }23\text{h~}$& $152$ & $11\text{h } 40$ & $38$ &$58$ $(68)$& $3\text{h } 25$&$38$& $\mathbf{170}$ $(204)$\\
     & 256$\times$256$\times$10$\times$2 &  $\sim 7\text{d~~}8\text{h}^*$ & $256$& $12\text{h } 20$ &$64$&$57$ $(68)$ & $5\text{h } 20$ &$64$& $\mathbf{173}$ $(204)$\\ \cmidrule(l){2-10}
     IST &152$\times$152$\times$17$\times$3 &  $3\text{d~}11\text{h~}$ &$152$& $5\text{h } 30$ &$38$&$60$ $(68)$& $1\text{h } 37$ &$38$& $\mathbf{184}$ $(204)$\\ \cmidrule(l){2-10}
     ALOHA  &152$\times$152$\times$17$\times$3 &   $\sim 25\text{d~~}1\text{h}^*$&$152$& $1\text{d }14\text{h } 25$ &$38$ &$62$ $(68)$& $18\text{h } 13$&$38$& $\mathbf{133}$ $(204)$\\ 
     \bottomrule
  \end{tabular}
  \vspace{-.4cm}
\end{table*}

\begin{table}[!t]
\centering
  \caption{Comparison of the learning-based random variable-density Gaussian sampling optimization for different settings. $n_{\text{pars}}$ denotes the size of the grid used to optimize the parameters. For each set of parameters, the results were averaged on $20$ masks drawn at random from the distribution considered. The $n_{\text{pars}}$ include a grid made of $12$ sampling rates (uniformly spread in $[0.025,0.3]$), $10$ different low frequency phase encodes (from $2$ to $18$ lines), and different widths of the Gaussian density (uniformly spread in $[0.05, 0.3]$) -- $10$ for the images of size $152\times 152$, 
  $20$ in the other case. }\label{tab:LB-VD}
  \begin{tabular}{>{\bfseries}c cccr}
  	\toprule
	Algo. & Setting  & $n_{\text{pars}}$  & $n_{\text{procs}}$ & Time \\
	\midrule
	  \multirow{2}{*}{KTF}  	& 152$\times$152$\times$17$\times$3~~ & 1200  &38 & $6$h~$30$ \\
	  					& 256$\times$256$\times$10$\times$2$^*$ & 2400  &64 & $6$h~$45$ \\ \cmidrule(l){2-5}
	IST  					& 152$\times$152$\times$17$\times$3~~ & 1200 &38 & $3$h~$20$\\ \cmidrule(l){2-5}
	ALOHA 				& 152$\times$152$\times$17$\times$3~~ & 1200  &38 & $1$d~$8$h\\
	\bottomrule
  \end{tabular}
  \vspace{0.2cm}
\end{table}

 \vspace{-.4cm}
\section{Multicoil experiments}
 \vspace{-.4cm}
For the multicoil experiment, we used the previously described cardiac dataset but we did not crop the images. We took the first $12$ frames for all subjects, and selected $4$ coils that cover the region of interest. Each image was then normalized in order for the resulting sum-of-squares image to have at most unit intensity.  When required, the coil sensitivities were self-calibrated according to the idea proposed in \citeapndx{feng2013highly}, which averages the signal acquired over time in the k-space and subsequently performs adaptive coil combination \citeapndx{walsh2000adaptive,griswold2002use}.

The advantage of using self-calibration is that the greedy optimization procedure can simultaneously take into account the need for accurate coil estimation as well as accurate reconstruction, thus potentially eliminating the need for a calibration scan prior to the acquisition. A more complete discussion of the accuracy of self-calibrated coil sensitivities is presented in \citeapndx{feng2013highly}.

We used $k$-$t$ SPARSE-SENSE  \citeapndx{otazo2010combination} and ALOHA \cite{jin2016general} for reconstruction. While the first requires coil sensitivities, the second reconstructs the images directly in k-space before combining the reconstructed data. We also introduce an additional mask designing baseline, namely golden ratio Cartesian sampling \citeapndx{li2018dynamic} that we will use in the sequel. We will refer to it as \textbf{golden}.

\begin{figure}[!ht]
  \centering
\hspace{-.5cm}\includegraphics[width=.95\linewidth]{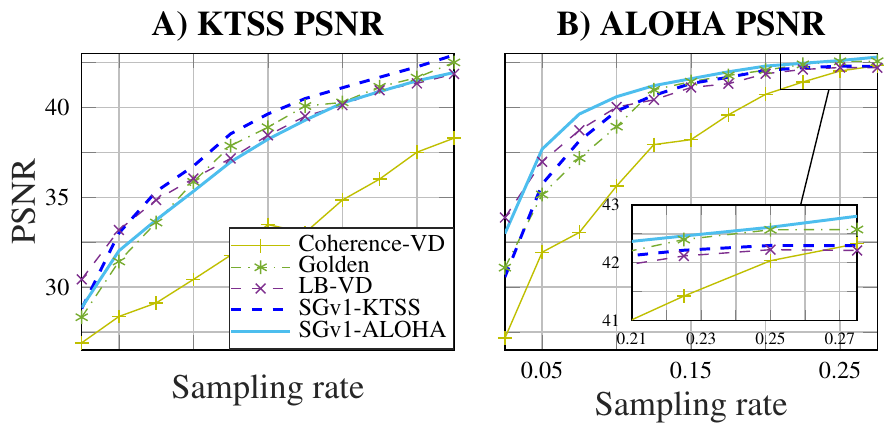}
 \vspace{-.4cm}
 \caption{\textbf{A-B)} PSNR as a function of sampling rate for KTSS \protect\citeapndx{otazo2010combination} and ALOHA \cite{jin2016general} in the multicoil setting, comparing SG-v1 with the coherence-VD \cite{lustig2007sparse}, LB-VD  and golden ratio Cartesian sampling \protect\citeapndx{li2018dynamic}, averaged on 4 testing images of size 256$\times$256$\times$12 with 4 coils. }\label{fig:psnr_multi}
  \vspace{-.4cm}
\end{figure}

\begin{figure}[!ht]
 \vspace{-.4cm}
\centering
\includegraphics[width=.85\linewidth]{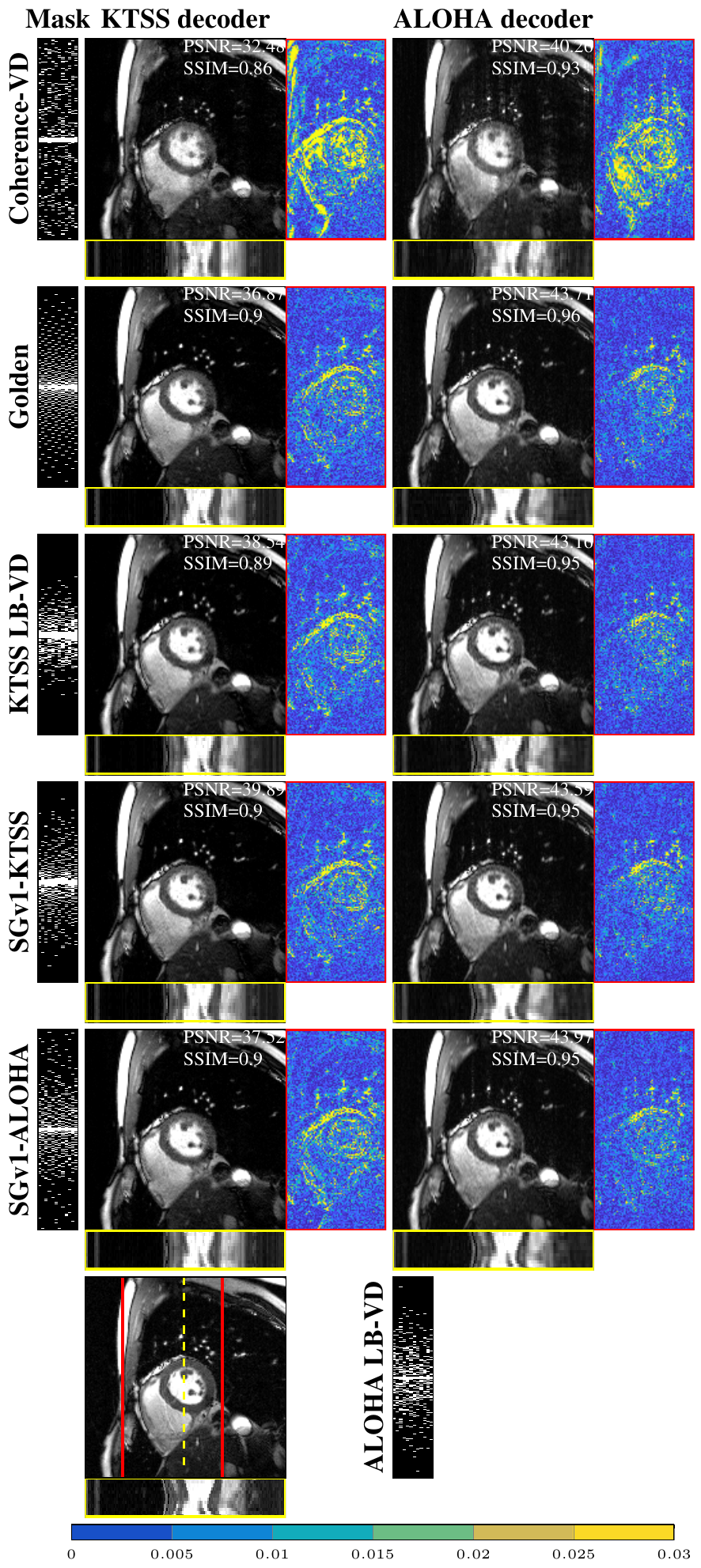}

 \caption{Reconstruction with KTSS \protect\citeapndx{otazo2010combination} and ALOHA \cite{jin2016general} at $15\%$  sampling rate for a 4 coil parallel acquisition of cardiac cine size 256$\times$256$\times$12. The setting is otherwise similar as the one presented in Figure 5 of \cite{sanchez2018}.}\label{fig:multi}
 \vspace{-.4cm}
\end{figure}

 \vspace{-.4cm}
\section{Additional single-coil results with SG-v1}\label{sec:single}
While the main paper focused on SG-v2, using a batch of training samples instead of the whole training set, we focus here on results with SG-v1. SG-v1 accelerated G-v1 by a factor $60$, and we contend that due to the small dataset used in our case, using a batch of training data instead of the whole set should not affect the performance. 
\vspace{-.4cm}
\subsection{Comparison to baselines}\label{sec:baseline}

The comparison to baselines is shown on Figures \ref{fig:psnr_cross} and \ref{fig:plot_cross}, where we see that the learning-based method yields masks which consistently improve the results compared to all variable-density methods used. Even though some variable-density techniques are able to provide good results for some sampling rates and algorithms, our learning-based technique is able to consistently provide improvement over this baseline. Compared to Coherence-VD, there is always at least $1$ dB improvement at any sampling rate, and it can be as much as $6.7$ dB at $5\%$ sampling rate for ALOHA. For \textit{golden}, there is an improvement larger than $1.5$ dB prior to $15\%$ rate, and around $0.5$dB after for all decoders. Figure \ref{fig:psnr_cross} also clearly indicates that the benefits of our learning-based framework become more apparent towards higher sampling rates, where the performance improvement over LB-VD reaches up to $1$ dB. Towards lower sampling rates, with much fewer degrees of freedom for mask design, the greedy method and LB-VD yield similar performance as expected. As shown in Figure \ref{fig:plot_cross}, the learning-based masks tend to conserve better the sharp contrast transition compared to the variable-density techniques. 

\vspace{-.4cm}
\subsection{Cross-performances of performance measures}
Up to here, we used PSNR as the performance measure, and we now compare it with the results of the greedy algorithm paired with SSIM, a metric that more closely reflect perceptual similarity. For brevity, we only consider ALOHA in this section. In the case where we optimized for SSIM, we noticed that unless a low-frequency initial mask is given, the reconstruction quality would mostly stagnate. This is why we chose to start the greedy algorithm with $4$ low-frequency phase encodes at each frame in the SSIM case.

The reconstructions for PSNR and SSIM are shown on Figure \ref{fig:plot_ssim}, where we see that the learning-based masks outperform the baselines across all sampling rates except at $2.5\%$ in the SSIM case. The quality of the results is very close for both masks, but each tends to perform slightly better with the performance metric for which it was trained. The fact that the ALOHA-SSIM result at $2.5\%$ has a very low SSIM is due to the fact that we impose $4$ phase encodes across all frames, and the resulting sampling mask at $2.5\%$ is a low pass mask in this case.

A visual reconstruction is provided in Figure \ref{fig:ssim}, we see that there is almost no difference in reconstruction quality, and that the masks remain very similar. Overall, we observe in this case that the performance metric selection does not have a dramatic effect on the quality of reconstruction, and our greedy framework is still able to produce masks that outperform the baselines when optimizing SSIM instead of PSNR.

\begin{figure}[t]
\centering
\vspace{-.2cm}
\includegraphics[width=.95\linewidth]{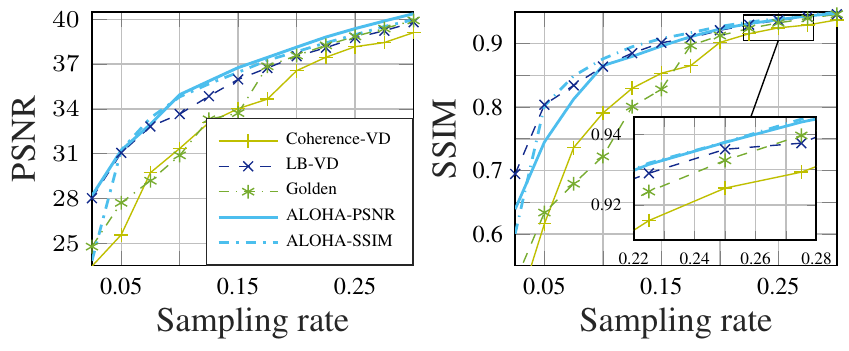}

\vspace{-.4cm}

\caption{PSNR and SSIM as a function of sampling rate for ALOHA, comparing the SG-v1 results optimized for PSNR and SSIM with the three baselines, averaged on $4$ testing images of size 152$\times$152$\times$17.}\label{fig:plot_ssim}

\includegraphics[width=.95\linewidth]{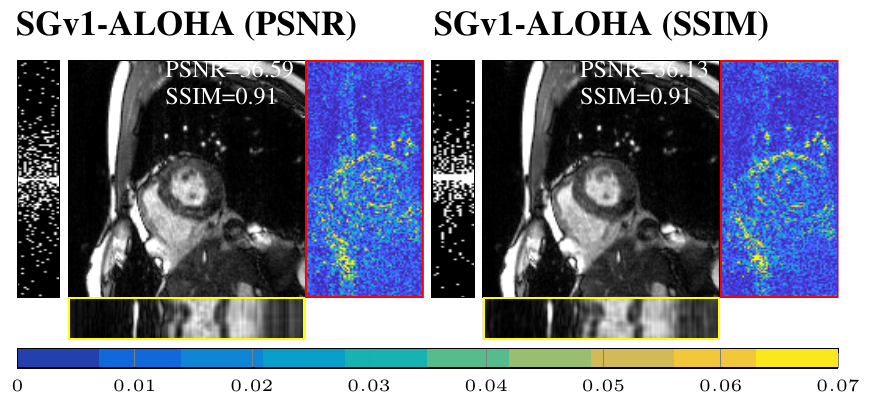}
\vspace{-.3cm}
\caption{Comparison of the sampling masks optimized for PSNR and SSIM with ALOHA, at $15\%$ sampling. The images and masks can be compared to those of Figure \ref{fig:plot_cross}, as the settings are the same.} \label{fig:ssim}
\end{figure}

\vspace{-.2cm}
\subsection{Experiments with different anatomies}\label{sec:large}

\begin{figure}[!t]
\centering
\includegraphics[width=.95\linewidth]{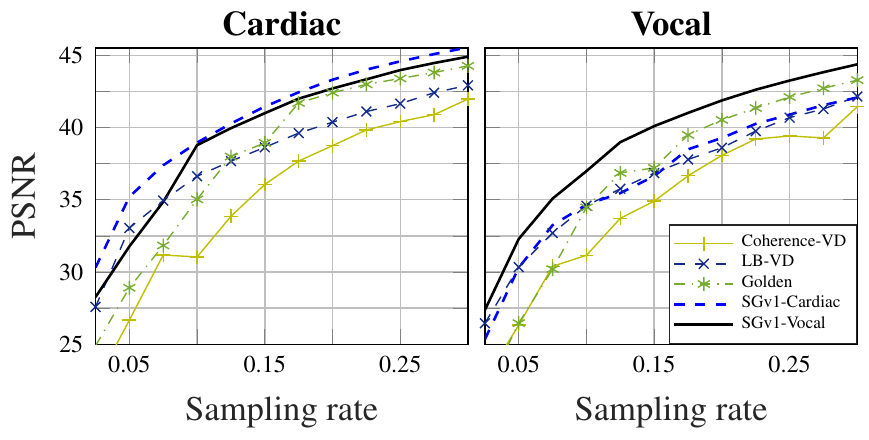}
\vspace{-.4cm}
\caption{PSNR as a function of sampling rate for KTF, comparing SG-v1 with both baselines, averaged on $2$ testing images for both cardiac and vocal data sets of size 256$\times$256$\times$10.}\label{fig:psnr_anat}
 \vspace{-.4cm}
\end{figure}

\begin{figure}[!t]
\centering
\vspace{.2cm}
\includegraphics[width=.95\linewidth]{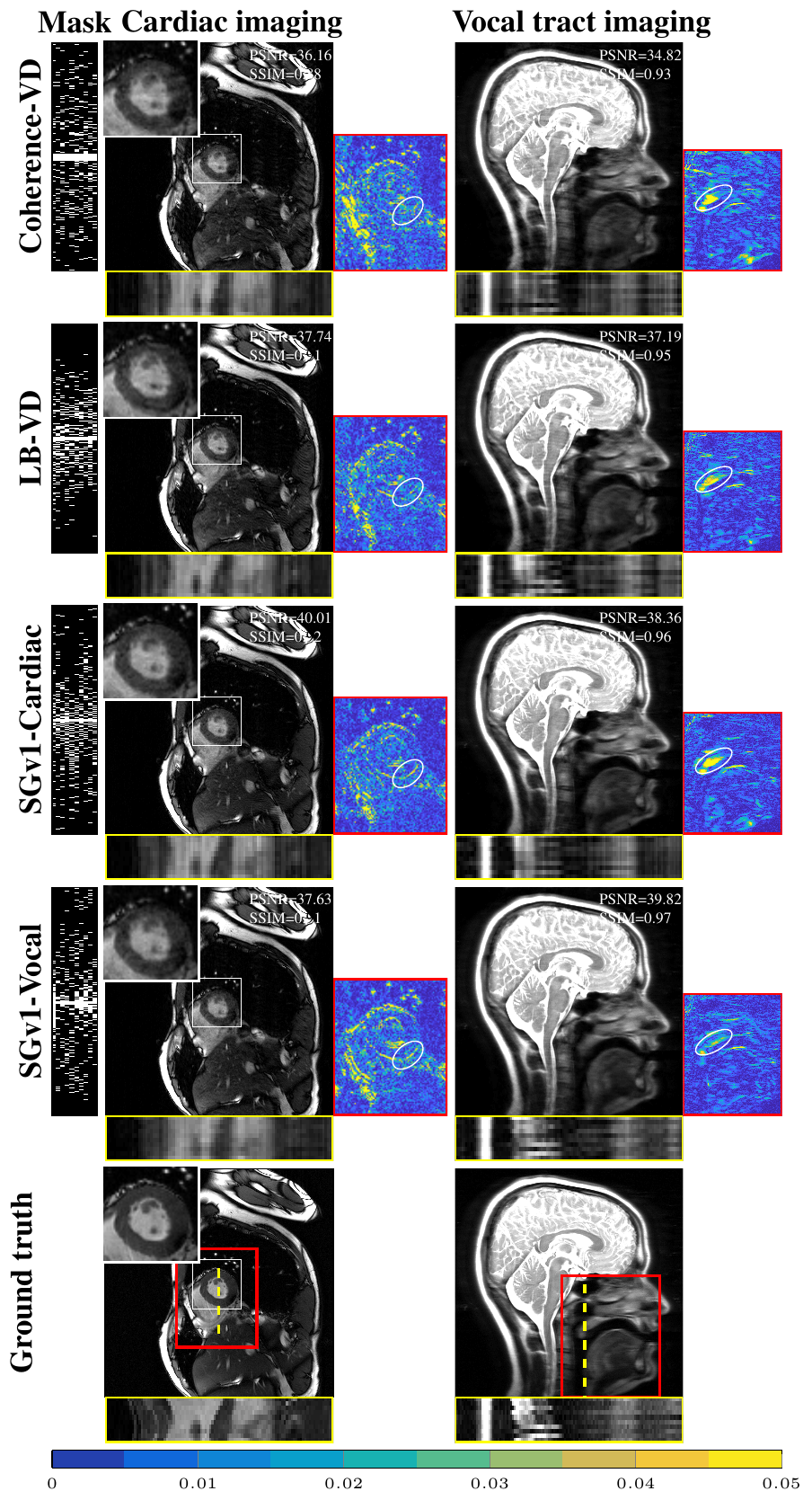}
\vspace{-.3cm}
\caption{Reconstruction for KTF at $15\%$ sampling for the cardiac and vocal anatomies of size 256$\times$256$\times$10. Figures showing different frames for the vocal and cardiac images are available in Figures \ref{fig:cardiac} and \ref{fig:vocal}. }\label{fig:plot_cross_anat}
\vspace{-.3cm}
\end{figure}

In these last experiments, we consider both the single coil cardiac dataset as well as the vocal imaging dataset both of size  $256\times 256 \times 10$. The cardiac dataset was trained on $5$ samples and tested on $2$, using only the first ten frames of each scan, whereas the vocal one used $2$ training samples and $2$ testing samples. In this setup, the k-space of the cardiac dataset tends to vary more from one sample to another than the vocal one, making the generalization of the mask more complicated. This issue would require more training samples, but imposing SG-v1 algorithm to start with $4$ central phase encoding lines on each frame was found to be sufficient to acquire the peaks in the k-space across the whole dataset.  \textit{SGv1-Cardiac} refers to the greedy algorithm using cardiac data, and \textit{SGv1-Vocal} is its vocal counterpart. The algorithm used a batch of size $k = 64$ at each iteration, and the results were obtained using only KTF.

The results are reported on the Figures \ref{fig:psnr_anat} and \ref{fig:plot_cross_anat}, and we see that, for the both datasets, the greedy approach provides superior results against VD sampling methods across all sampling rates. It is striking that, in this setting, the SG-v1 approach outperforms even more convincingly all the baselines, and the LB-VD approach, in this case, is outperformed by more than $2$dB by SG-v1, where it remained very competitive in the other settings. This difference is clear in the temporal fidelity of both reconstructions on Figure \ref{fig:plot_cross_anat}, where we see that the LB-VD approach loses sharpness and accuracy compared to SG-v1.

\vspace{-.5cm}
\subsection{Comparison across anatomies}\label{sec:anatomies}

The main complication coming from applying the masks across anatomies is that the form of the k-space might vary heavily across datasets: the vocal spectrum is very sharply peaked, while the cardiac one is much broader. 
Comparing the cross-performances on Figures \ref{fig:plot_cross_anat}, we see that the and \textit{SGv1-vocal} masks generalizes much better on the cardiac datasets than the other way around. This can be explained from the differences in the spectra: the cardiac one being more spread out, the cardiac mask less faithfully captures  the very low frequencies of the k-space, which are absolutely crucial to a successful reconstruction on the vocal dataset, thus hindering the reconstruction quality. Also, we see that it is important for the trained mask to be paired with its anatomy to obtain the best performance.

 \vspace{-.5cm}
\subsection{Additional visual reconstructions for cardiac and vocal dataset}\label{app:visual}

The present appendix provides further results for experiments \ref{sec:large} and \ref{sec:anatomies}. We show in Figures \ref{fig:cardiac} and \ref{fig:vocal} reconstruction at different frames which provide clearer visual information to the quality of reconstruction compared to the temporal profiles.
\begin{figure*}[!h]
\vspace{-.5cm}
\centering
\begin{minipage}[c]{.65\linewidth}
	\vspace{0pt}
	\includegraphics[width=1.1\linewidth]{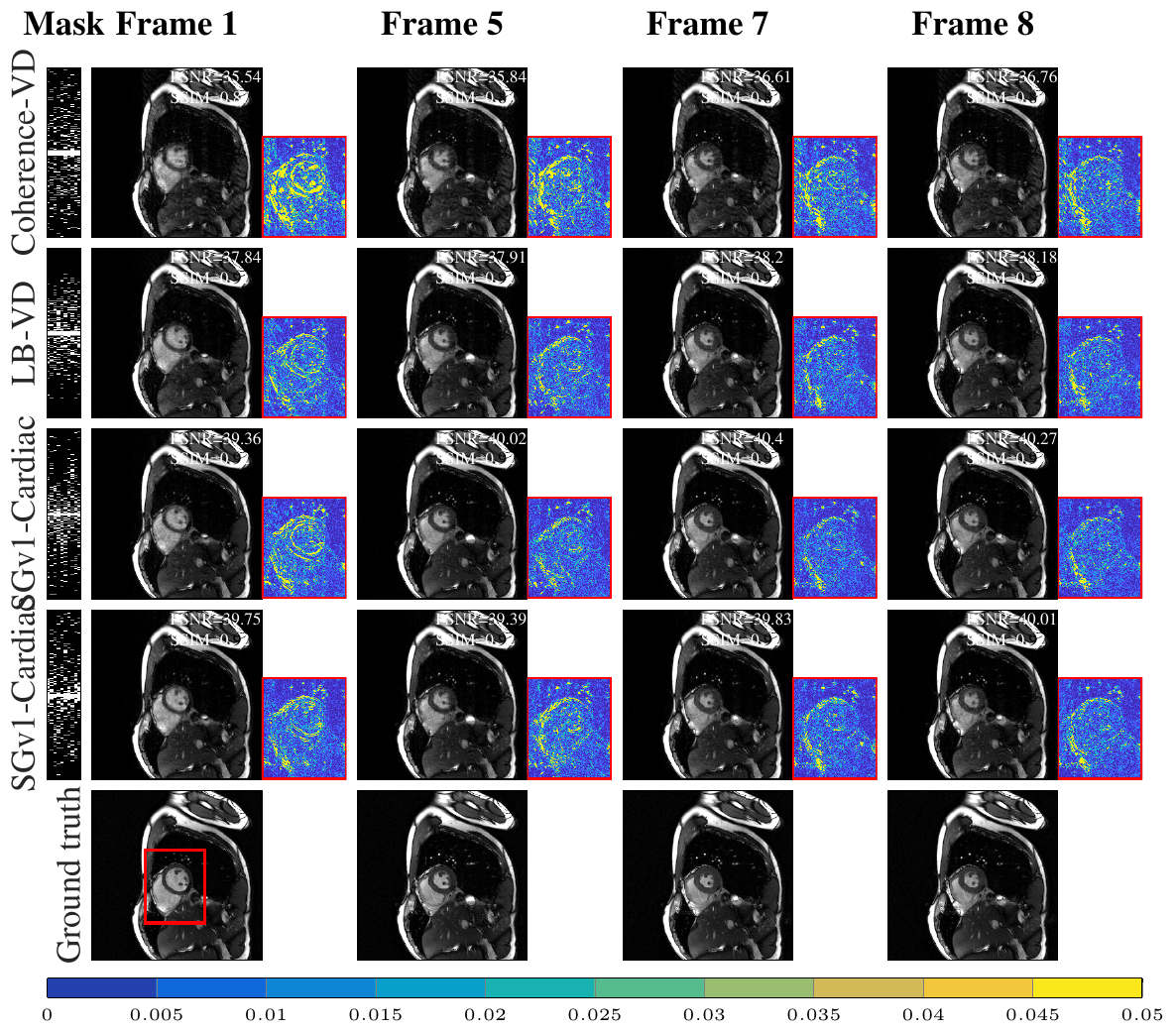}
\end{minipage}\hfill
\begin{minipage}[c]{.25\linewidth}
	\centering \caption{Reconstruction with KTF \cite{jung2009k} at $15\%$ sampling rate for the cardiac anatomy of size 256$\times$256$\times$10. It unfolds the temporal profile of Figure \ref{fig:plot_cross_anat}. The PSNR and SSIM displayed are computed for a each image individually, and the overall PSNR for each image is the one of Figure \ref{fig:plot_cross_anat}. The ground truth is added at the end of each line for comparison.}\label{fig:cardiac}
\end{minipage}
\begin{minipage}[c]{.65\linewidth}
	\vspace{0pt}
	\includegraphics[width=1.1\linewidth]{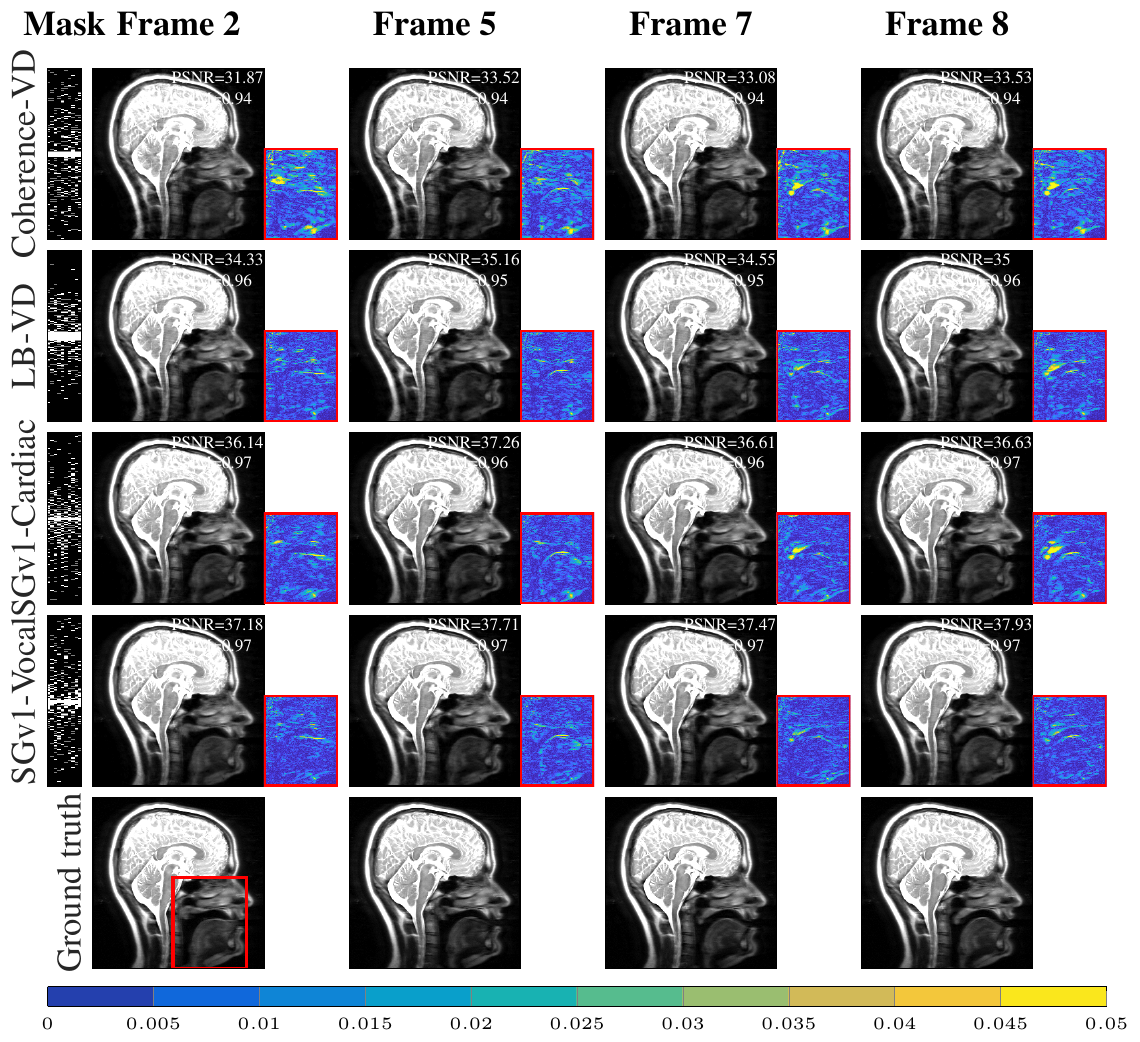}
\end{minipage}\hfill
\begin{minipage}[c]{.25\linewidth}
	\centering \caption{Reconstruction with KTF at $15\%$  \cite{jung2009k} sampling rate for the vocal anatomy of size 256$\times$256$\times$10. It unfolds the temporal profile of Figure \ref{fig:plot_cross_anat}. The PSNR and SSIM displayed are computed for a each image individually, and the overall PSNR for each image is the one of Figure \ref{fig:plot_cross_anat}. The ground truth is added at the end of each line for comparison. }\label{fig:vocal}
\end{minipage}
	\vspace{-.5cm}

\end{figure*}

For these images, the PSNR and SSIM are computed with respect to each individual frame, showing the quality of the reconstruction in a much more detailed fashion than before, where we considered each dynamic scan as a whole. Generally, we as previously observed, the mask trained for a specific anatomy will most faithfully capture the sharp contrast transitions in the dynamic regions of the images. For the vocal images, we see that sampling the first frame more heavily is important in order to avoid having a very large PSNR discrepancy, as observed for the other masks.  The PSNR remains quite stable across the frames otherwise.

\newpage
\subsection{Noisy experiments}\vspace{-.2cm}

In order to test the robustness of our framework to noise, we artificially added bivariate circularly symmetry complex random Gaussian noise to the normalized complex images, with a standard deviation $\sigma = 0.05$ for both the real and imaginary components. We then tested to see whether the greedy framework is able to adapt to the level of noise by prescribing a different sampling pattern than in the previous experiments.

  We chose to use V-BM4D \citeapndx{maggioni2012video} as denoiser with its default suggested mode using Wiener filtering and low-complexity profile, and provided the algorithm the standard deviation of the noise as the denoising parameter. The comparison between the fully sampled denoised images and the original ones yields an average PSNR of $24.95$ dB across the whole dataset. Due to the fact that none of the reconstruction algorithms that we used have a denoising parameter incorporated, we simply apply the V-BM4D respectively to the real and the imaginary parts of the result of the reconstruction. The results that we obtain are presented on the Figures \ref{fig:psnr_noisy} and \ref{fig:plot_noisy}. 

It is interesting to notice on Figure \ref{fig:plot_noisy} that the learning-based framework outperforms the baselines that are not learning-based by a larger margin than in the noiseless case, and this is again especially true at low sampling rates. In this case however, the difference between SG-v1 and LB-VD methods is much smaller, and this might be explained by the fact that noise corrupts the high frequency samples, and thus the masks concentrate more around low-frequencies, leaving less room for designs that largely differ.

We see a clear adaptation of the resulting learning based mask, as shown by comparing Figures \ref{fig:plot_cross} and \ref{fig:plot_noisy}: the masks SGv1-KTF and SGv1-ALOHA, which are trained on the noisy data, are closer to low-pass masks, due to the high-frequency details being lost to noise, and hence, no very high frequency samples are added to the mask. 

 Also, notice than even if the discrepancy in PSNR is only around $0.8-1$ dB between the golden ratio sampling and the optimized one, the temporal details are much more faithfully preserved by the learning-based approach, which is crucial in dynamic applications. The inadequacy of coherence-based sampling is highlighted in this case, as very little temporal information is captured in the reconstruction with both decoders. Also, for both decoders, there is a clear improvement on the preservation of the temporal profile when using learning-based masks compared to the baselines; the improvement of the SGv1-ALOHA mask of around $3$dB also shows how well our framework is able to adapt to this noisy situation, whereas Coherence-VD yields results of unacceptable quality.

\begin{figure}[!t]
\centering
\includegraphics[width=.9\linewidth]{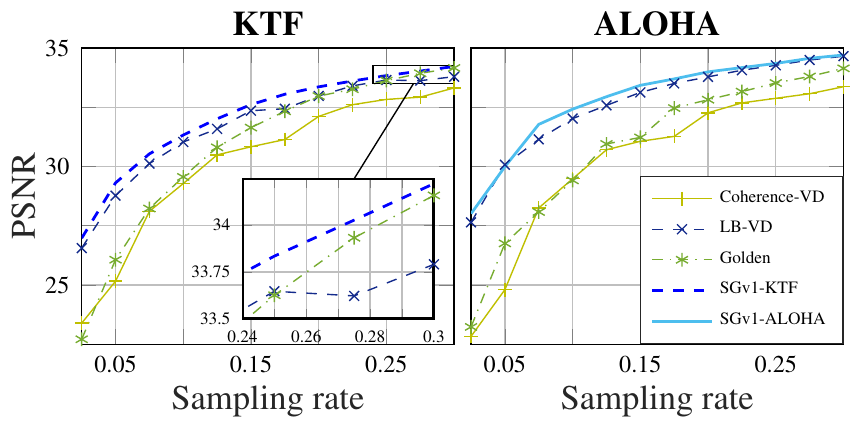}
\vspace{-.4cm}
\caption{PSNR as a function of sampling rate for both reconstruction algorithms considered, comparing SG-v1 with the three baselines, averaged on $4$ noisy testing images of size 152$\times$152$\times$17. The PSNR is computed between the denoised reconstructed image and the original (not noisy) ground truth.}\label{fig:psnr_noisy}

\includegraphics[width=.9\linewidth]{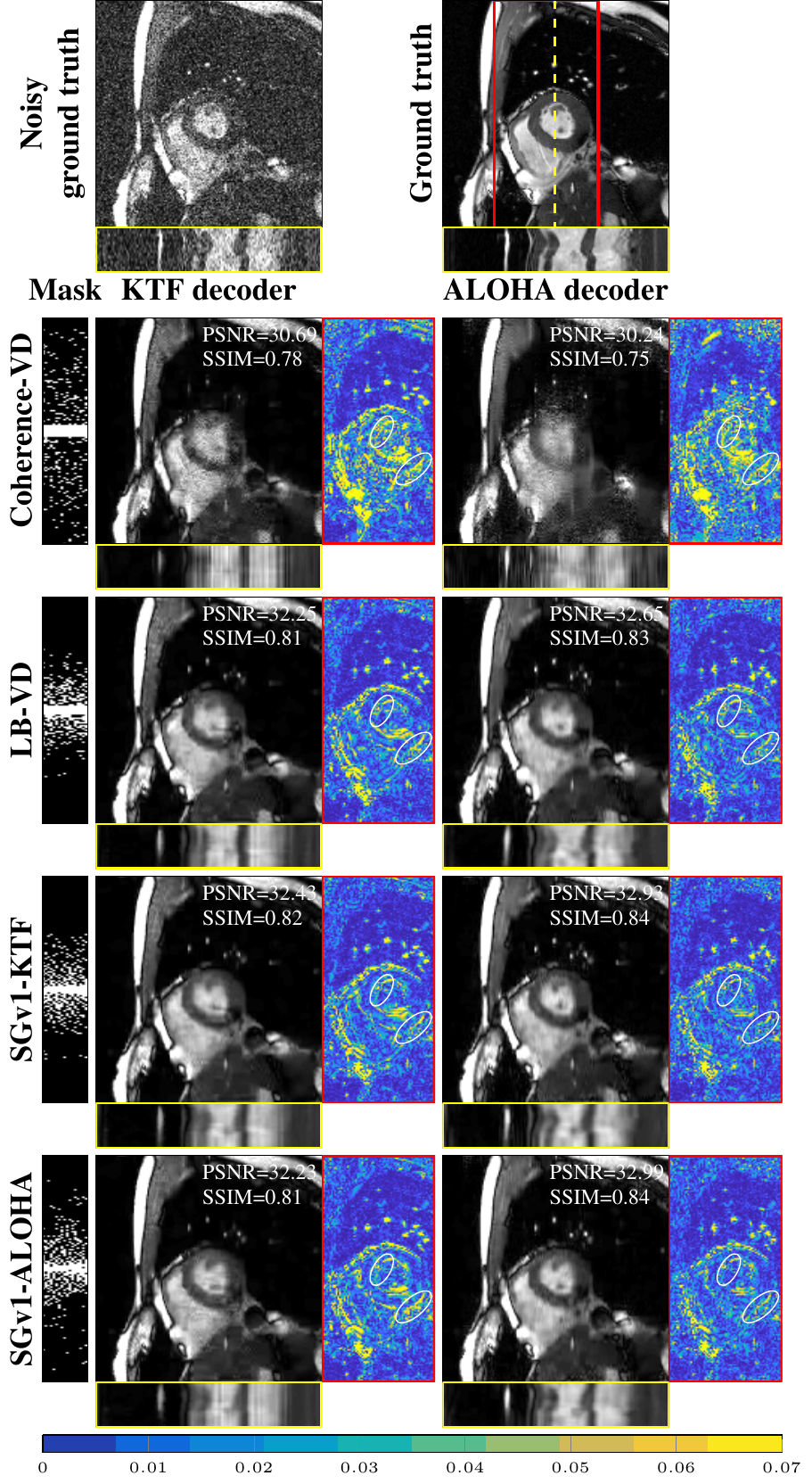}
\vspace{-.4cm}
\caption{Reconstructed denoised version from the noisy ground truth on the first line, at $15\%$ sampling. The PSNR is computed with respect to the original ground truth on the top right.}\label{fig:plot_noisy}
\end{figure}
\newpage

\bibliographystyleapndx{IEEEtran}
\bibliographyapndx{biblio}
\end{document}